\documentclass[12pt]{article}

\usepackage{setspace}
\usepackage[
    a4paper,
    left=0.9in,
    right=0.9in,
    top=0.5in,
    bottom=1in
]{geometry}

\usepackage[utf8]{inputenc}
\usepackage[T1]{fontenc}
\usepackage{lmodern}

\usepackage{amsmath}
\usepackage{amsthm}
\usepackage{amssymb}
\usepackage{amsfonts}
\usepackage{mathtools}
\usepackage{bm}
\usepackage{bbm}

\usepackage{graphicx}
\usepackage{subcaption}
\usepackage{float}
\usepackage{booktabs}
\usepackage{multirow}
\usepackage{array}
\usepackage{siunitx}

\usepackage{algorithm}
\usepackage{algpseudocode}

\usepackage{tikz}
\usetikzlibrary{positioning,arrows.meta}

\usepackage{xcolor}
\usepackage{url}
\usepackage{csquotes}
\usepackage[shortlabels]{enumitem}
\usepackage{comment}
\usepackage{tcolorbox}

\usepackage[numbers,sort&compress]{natbib}

\usepackage{hyperref}

\hypersetup{
    colorlinks=true,
    linkcolor=blue,
    filecolor=magenta,
    urlcolor=cyan,
    citecolor=blue
}

\newcommand\numberthis{%
    \addtocounter{equation}{1}%
    \tag{\theequation}%
}

\newtheorem{theorem}{Theorem}
\newtheorem{lemma}[theorem]{Lemma}
\newtheorem{corollary}[theorem]{Corollary}

\theoremstyle{definition}

\theoremstyle{remark}

\newcommand{\R}{\mathbb{R}}
\newcommand{\X}{\mathcal{X}}
\newcommand{\bbS}{\mathbb{S}}
\newcommand{\balpha}{\boldsymbol{\alpha}}
\newcommand{\bbeta}{\boldsymbol{\beta}}

\DeclareMathOperator{\tr}{\text{tr}}

\title{
    Quantification and Decomposition of Uncertainty Using
    Sliced-Normal Distribution -- With Applications to NASA Data
}

\author{
Arindam RoyChowdhury%
\thanks{
Industrial Engineering and Operations Research, Columbia University,
\href{mailto:arindam.roychowdhury@columbia.edu}
{ar4445@columbia.edu}
}
\and
Luis G. Crespo%
\thanks{
NASA Langley Research Center, Vehicle Dynamics \& Controls Branch,
\href{mailto:luis.g.crespo@nasa.gov}
{luis.g.crespo@nasa.gov}
}
\and
Henry Lam%
\thanks{
Industrial Engineering and Operations Research, Columbia University,
\href{mailto:henry.lam@columbia.edu}
{henry.lam@columbia.edu}
}
}

\date{}

\begin{document}

\maketitle

\hrule

\begin{abstract}
Modeling multivariate distributions with nonlinear dependence, multimodality, and tractable analytical structure for downstream applications is a central challenge in uncertainty quantification. Sliced Normal (SN) distributions were introduced in prior works at the National Aeronautics and Space Administration (NASA) to address this need by representing densities through polynomial feature maps. This construction provides a compact algebraic alternative to more opaque generative models, while retaining the ability to capture nonlinear parameter dependencies and multi-modal behavior. In this paper, we build on the SN framework and develop several improvements that make the approach more reliable and scalable. First, we reformulate SN parameter estimation as a convex optimization problem over a positive semidefinite matrix, replacing the original nonconvex likelihood search with a formulation amenable to standard optimization tools. Second, we clarify the expressive power of the SN class by connecting polynomial log-density modeling to a Stone--Weierstrass-type universal approximation argument on compact domains. Third, we propose a high-dimensional fitting procedure that partitions variables into approximately independent groups, fits SN models within each subgroup, and then assembles the subgroup models through a cross-block completion step to recover residual dependence. We demonstrate the resulting SN modeling pipeline on NASA loss-of-control flight data, where the method captures nonlinear dependence patterns in both low-dimensional slices and a higher-dimensional block-assembled model.

\end{abstract}

\vspace{6pt}


\vspace{10pt}


\section{Introduction}\label{sec:Intro}

A central challenge in data science problems is to build probabilistic models that are not only flexible enough to describe complex data, but also structured enough to support downstream analysis, optimization, and decision-making. To address these challenges, Sliced Normal (SN) distributions have been developed at the National Aeronautics and Space Administration (NASA) and used as a practical modeling tool for uncertainty quantification \citep{Crespo_SN_First, sliced_exponential, SN_application}. These applications require more than a black-box density estimate: the fitted distribution must remain interpretable, computationally tractable, and useful for tasks such as identifying high-density regions, failure regions, most probable points, and uncertainty decompositions.

An important example where this modeling need arises is the NASA Langley flight dataset studied in this paper, which comes from experiments related to in-flight loss of control (LOC). LOC is the largest fatal accident category for commercial jet airplane accidents worldwide \citep{Belcastro11}. Broadly, LOC refers to aircraft motion outside the normal operating flight envelope that is not predictably altered by pilot commands, and is often driven by strong nonlinear effects and coupling \citep{Belcastro11,Crespo12}. Such events are frequently characterized by disproportionately large responses to small perturbations in the vehicle state, including oscillatory or divergent behavior. In particular, the uncommanded angular rates that emerge in these situations can severely compromise the ability to maintain heading, altitude, and wings-level flight. To study this phenomenon, flight experiments were conducted using the Generic Transport Model (GTM), a \(5.5\%\) dynamically scaled, remotely piloted, twin-turbine aircraft. The dataset used here comes from flights conducted under critical upset conditions. Although these flights are nominally identical, the recorded responses exhibit substantial variability, reflecting the complex and nonlinear nature of the underlying dynamics. 

These features motivate the use of the SN framework, which provides a compact and tractable representation for modeling nonlinear dependence in multivariate data. The appeal of this framework becomes clearer when viewed against the limitations of existing distributional modeling approaches. Copula-based methods decompose a joint distribution into marginal models and a dependence model, and structured variants such as vine copulas can provide substantial flexibility in moderate dimensions \citep{nelsen2006introduction}. In higher dimensions, structured constructions like vine copulas—including C-vine and R-vine formulations—offer flexible tools for building multivariate models from simpler pairwise relationships \citep{aas2009pair, czado2019analyzing}. Despite their flexibility, the performance of copula models is highly sensitive to the choice of copula family. Moreover, standard families often fall short in accurately capturing nonlinear or higher-order interactions, limiting their expressiveness in complex, real-world applications. Recent advances in generative modeling have introduced highly flexible approaches for representing complex, high-dimensional distributions. In particular, Energy Based Models (EBMs) and Normalizing Flows have emerged as powerful tools. EBMs specify a density through an unnormalized form
\[
p_\theta(x) \propto \exp\!\big(-E_\theta(x)\big),
\]
where a neural network typically parameterizes the energy function $E_\theta(x)$ \citep{lecun2006tutorial,song2021score}. Normalizing flows, on the other hand, construct expressive densities by transforming a simple base distribution through a sequence of invertible mappings, allowing exact likelihood evaluation via change-of-variables formulas \citep{rezende2016variationalinferencenormalizingflows,papamakarios2021normalizing}. Both frameworks are capable of modeling highly nonlinear, multimodal, and high-dimensional data distributions. However, these neural generative models generally produce distributions that are \emph{mathematically intractable}. Their density landscapes are defined implicitly through deep networks, making it difficult to analyze structural properties such as the number or location of modes. Moreover, optimization tasks such as computing the most probable point (MPP) or solving downstream polynomial optimization problems under these distributions become analytically infeasible. Even basic probabilistic quantities often require costly sampling-based approximations.

To overcome these limitations, SN distributions were introduced by \cite{Crespo_SN_First} as a flexible parametric family capable of approximating complex multivariate densities. The key idea behind SN distributions is to map original data into a higher-dimensional feature space using polynomial transformations. In this lifted space, a quadratic form is used to construct a log-density, thereby generalizing Gaussian models defined over polynomial features. This construction allows SN distributions to capture multimodal structures, asymmetric shapes, and potentially non-linear dependencies using a low-dimensional representation in the feature space. Unlike kernel density estimation or Gaussian mixture models, using the original formulation of SNs do not require manual specification of component counts or bandwidths, and their compact polynomial representation is particularly well-suited for downstream activities like identifying level sets and failure regions.

SN distributions provide a structured alternative within the broader class of energy-based models. In SN models, the energy takes the form
\[
E_\theta(x) = \sum_{k=1}^m \big(f_k(x)\big)^2,
\]
where each $f_k$ is a polynomial. Thus, the log-density is a negative sum of squared polynomials. This algebraic structure provides mathematical tractability absent in neural EBMs: the energy landscape is polynomials, enabling the use of tools from polynomial optimization and algebraic analysis. In particular, gradients and higher-order derivatives are available in closed form, allowing calibration via score matching using analytical gradients rather than stochastic approximations.

Furthermore, once an SN model is estimated, its polynomial parameterization permits studying how the distribution evolves as parameters vary continuously. This enables sensitivity analysis and exploration of parameter dependence in a continuum—an analysis that is generally infeasible for neural generative models whose structure is implicit and non-algebraic.

However, learning SN parameters poses significant computational challenges due to the non-convex nature of the underlying optimization problem, which involves maximizing a log-likelihood over the cone of positive semi-definite matrices. Moreover, the interpretability and identifiability of parameters become more fragile as the degree of the feature mapping increases.

In this work, we introduce several key advancements to the theory and application of SN distributions, leveraging both optimization methodologies and statistical analysis:

\begin{enumerate}
\item \textbf{Convex Formulation:} Recognizing the computational challenges of the original non-convex parameter estimation problem, we propose a convex reformulation of the optimization task. This enables us to leverage highly optimised solvers like MANOPT or CVXPY for fast parameter estimation. 

\item \textbf{Consistency:} We justify the consistency of SN distributions by relating to the Stone–Weierstrass theorem, establishing the dense approximation capability of polynomial-based families for modeling continuous log-density functions on compact sets. 

\item \textbf{High-Dimensional Applications:} We introduce a practical strategy for extending SN models to high-dimensional settings. Our approach begins by identifying approximately independent variable groups within the dataset. We then fit separate SN models to each subgroup and subsequently combine them to construct a global model. This modular decomposition allows us to improve computational efficiency while still capturing the joint structure, and we further refine the combined model by optimizing its overall likelihood.
\end{enumerate}

Our contributions reinforce the theoretical underpinnings of SN distributions and enhance their computational tractability. The structure of the paper is as follows. In Section~\ref{sec: SN_recap}, we review the original formulation of SN models and summarize the key foundational ideas. Section~\ref{sec: Convex_reformulation} introduces a convex reformulation of the parameter estimation problem, along with algorithmic details for its implementation. Additionally, we analyze consistency properties of the SN model class. Section~\ref{sec: approximate_procedures} outlines practical approximate methods for fitting SN models in high-dimensional settings, including a dimension-reduction strategy based on identifying approximately independent subgroups. Section~\ref{sec: numerical_example} presents numerical experiments that illustrate the empirical benefits of our approach. Finally, Section~\ref{sec: limitaion} discusses known limitations of the SN framework and outlines directions for future work.

\section{Sliced Normal Distributions}\label{sec: SN_recap}
We begin by introducing the Sliced Normal (SN) distribution for a random variable \( X \) taking values in a compact subset \( \mathcal{X} \subset \mathbb{R}^m \). Let \( \tilde Z(x) \) represent the vector of all monomials of \( x \) of degree less than or equal to $d$. To formalize, for \( d \in \mathbb{N}_0 \) set 
\begin{equation}\label{eqn: set_of_powers}
    \Lambda_d := \{ (\alpha_1, \ldots, \alpha_m) \in \mathbb{N}_0^m \mid 1 \leq \alpha_1 + \cdots + \alpha_m \leq d \}
\end{equation}

Then, $\tilde Z(x)$ is the $|\Lambda_d|\times 1$ vector:
$$\tilde Z(x)^{\top} = (x^{\alpha^{(1)}},..., x^{\alpha^{(|\Lambda_d|)}})$$
Where $|\Lambda_d|$ denote the cardinality of the set $\Lambda_d$ and $\alpha^{(i)}$ are the elements of $\Lambda_d$. For instance, when \( m = 3 \) and \( d = 2 \), we have:
\begin{equation}\label{eq:feature_definition}
    x^\top = (x_1, x_2, x_3) \in \mathbb{R}^3 \quad \text{and} \quad \tilde Z(x)^\top = (x_1^2, x_2^2, x_3^2, x_1x_2, x_1x_3, x_2x_3, x_1, x_2, x_3) \in \mathbb{R}^{9}.
\end{equation}

The range of \( \tilde Z(x) \), denoted as \( \mathcal{\tilde Z} \), is defined as:
\[
\mathcal{\tilde Z} := \{\tilde Z(x) : x \in \mathcal{X}\}.
\]
Thus, the mapping \( \tilde Z(\cdot) \colon \mathcal{X} \to \mathcal{\tilde Z} \) satisfies \( \mathcal{\tilde Z} \subset \mathbb{R}^{m_z} \), where \( m_z = |\Lambda_d| =  \binom{m+d}{d} -1\). The quantity \( \tilde Z(x) \) is referred to as the \textbf{features} of \( x \), and the space \( \mathcal{\tilde Z} \) is called the \textbf{feature space}. For convenience, we will henceforth denote \( \tilde Z(x) \) as \( \tilde Z_x \).

Let \( f(x) \) denote the Probability Density Function (PDF) of the random variable \( X \). The SN distribution is parameterized by two quantities, \( (\mu, P) \). Here, $\mu$ is a $m_z\times 1$ vector and $P$ is a $m_z\times m_z$ positive definite matrix, called the precision matrix. We say \( X \sim \text{SN}(\mu, P) \) if:
\begin{equation}\label{eq:sn_pdf_old}
    f(x| \mu,P) = \frac{1}{c(\mu, P)} \exp\left(- \frac 12(\tilde Z_x - \mu)^\top P (\tilde Z_x - \mu)\right),
\end{equation}
where \( c(\mu, P) \) is a normalization constant ensuring \( \int_{\mathcal{X}} f(x|\mu,P) \, dx = 1 \).

Observe that the PDF in \eqref{eq:sn_pdf_old} closely resembles that of a multivariate normal distribution. Equivalently, it can be expressed as:
\[
f(x|\mu,P) \propto \phi_{\mu, P}(\tilde Z_x),
\]
where \( \phi_{\mu, P}(\tilde Z_x) \) is the PDF of a multivariate normal distribution with mean \( \mu \) and precision matrix \( P \), at the point $\tilde Z_x$.

\subsection{Parameter Estimation for Standard SNs}\label{sec: SN parameter estimation}
The task here is to estimate the parameters $\mu$ and $P$ given a data $X_1,...,X_n$. A very common approach to this, which is also used in \citep{Crespo_SN_First} is the Maximum Likelihood Estimation  (MLE). Where, the optimal parameters $\mu^*,P^*$ is denoted by:
$$(\mu^*,P^*):= \arg\max_{\mu,P} \ \prod_{i = 1}^n f(x_i|\mu,P)$$
Equivalently, it is customary to maximize the log likelihood, instead of the above:
\begin{align*}
    (\mu^*,P^*) &:= \arg\max_{\mu,P} \ \sum_{i = 1}^n  \log \left(f(x_i|\mu,P)\right)\\
    &= \arg\max_{\mu,P} \ -\frac 12\sum_{i = 1}^n  (\tilde Z_{x_i} - \mu)^\top P (\tilde Z_{x_i} - \mu) - n\log c(\mu,P)
\end{align*}
subject to the condition $P\succeq 0$.
Generally, a closed form solution to the above is not available, and numerical methods are employed. At the same time, it can be shown that the above optimization problem is not convex. Further, the optimization above requires evaluation of the integration constant $c(\mu,P)$, for which a closed form is unavailable. This is computed through Monte Carlo integration, a practice that is computationally expensive and limits scaling to higher dimensions. Thus, it is difficult to guarantee global optima. A proposed starting point for the numerical search of the above, also known as the ``Feature Space MLE'' or FMLE  is denoted by the moments of the features $\tilde Z_{x_i}$:
\begin{equation}\label{FMLE_mu}
    \hat\mu := \frac{1}{n}\sum_{i=1}^n \tilde Z_{x_i}
\end{equation}

and 
\begin{equation}\label{FMLE_P}
    \hat P := \left[\frac{1}{n-1}\sum_{i=1}^n(\tilde Z_{x_i}-\hat \mu)(\tilde Z_{x_i}-\hat \mu)^T\right]^{-1}
\end{equation}

\section{Convex Reformulation}\label{sec: Convex_reformulation}
The SN density function can be expressed as an exponential polynomial, i.e.,
\[
f(x) \propto e^{-g(x)},
\]
where \( g(x) \) is a polynomial in \( x \). However, we impose specific structural constraints on \( g(x) \). In particular, the positive definiteness of the precision matrix \( P \) ensures that \( g(x) \) is a sum of squares (SOS) polynomial. Next we explore a function class that includes all polynomials.

In a generalized setup, \( g(x) \) can be expressed as:
\[
g(x) = \theta_0 + \theta^\top \tilde Z_x,
\]
where \( \theta_0 \) is a scalar, and \( \theta \) is a vector of size \( m_z \times 1 \). Similar to the sliced exponential distribution introduced in \cite{sliced_exponential}, parameter estimation in this framework becomes a convex optimization problem. Consequently, the global optimum can be efficiently obtained using standard, well-established algorithms.

Despite its computational advantages, this generalized class of functions introduces significant challenges. A necessary condition for the integration constant to be finite is
\[
\lim_{\|x\| \to \infty} f(x) = 0.
\]
This property is not guaranteed if \( g(x) \) is an arbitrary polynomial. One might argue that this limitation is negligible because the domain \( \mathcal{X} \) is a compact subset of \( \mathbb{R}^m \). However, allowing arbitrary polynomials makes the model highly susceptible to outliers. Specifically, points near the boundary of \( \mathcal{X} \) can disproportionately inflate the estimated PDF in those regions, leading to severe non-robustness.

To mitigate these issues, one might consider restricting \( g(x) \) to be a positive polynomial. A multivariate polynomial is positive if it takes non-negative values over its domain. For example, a polynomial that is a square, or a non-negative linear combination of squares (SOS), is positive, and this representation serves as an algebraic certificate of its positivity. Unfortunately, Hilbert showed that a polynomial which is positive on all of real affine n-space, where $n > 1$, is not necessarily SOS. In fact, current literature does not provide any computationally tractable certificate of positive polynomials in higher dimensions. Given these constraints, we restrict \( g(x) \) to be an SOS polynomial. This restriction ensures computational tractability while providing a robust framework for parameter estimation and model fitting, due to the following lemma \ref{lemma: SOS_characterization}. Before proceeding, we slightly modify the notation introduced in section \ref{sec: SN_recap}. We redefine the set $\Lambda_d$ from eqn \ref{eqn: set_of_powers} to also contain the constant $\alpha = (0,...,0)$:
\begin{equation}\label{eqn: set_of_powers_new}
    \Lambda_d := \{ (\alpha_1, \ldots, \alpha_m) \in \mathbb{N}_0^m \mid  \alpha_1 + \cdots + \alpha_m \leq d \}
\end{equation}
Now, $m_z = |\Lambda_d| = \binom{m+d}{d}$. 
 Correspondingly, we also update $\tilde Z(x)$ to be the $|\Lambda_d|\times 1$ vector:
 \begin{equation}\label{eqn: feature_definition_new}
     Z(x)^{\top} = (x^{\alpha^{(1)}},..., x^{\alpha^{(|\Lambda_d|)}})
 \end{equation}

 Thus, $Z(x)$ now contains the constant term $1$. For example, for $m=3$ and $d=2$,  eqn. \ref{eq:feature_definition} now becomes:

 \begin{equation}
    x^\top = (x_1, x_2, x_3) \in \mathbb{R}^3 
\end{equation}
and
\begin{equation}
        Z(x)^\top = (1,x_1^2, x_2^2, x_3^2, x_1x_2, x_1x_3, x_2x_3, x_1, x_2, x_3) \in \mathbb{R}^{10}
\end{equation}

Then, we have the following lemma:
\begin{lemma}\label{lemma: SOS_characterization}    
    \citep{Powers1998} Suppose \( g \in \R[x] \) is a polynomial of degree \( 2k \) and \( Z_x \) is a vector of monomials of $x$ as in  (\ref{eqn: feature_definition_new}). Then \( g \) is a SOS polynomial if and only if there exists a real, symmetric, positive semidefinite matrix \( B \) such that
\[
g(x) = Z_x^\top B Z_x.
\]

Given such a matrix \( B \) of rank \( t \), we can construct polynomials \( h_1, \dots, h_t \) such that 
\[
g = \sum_{i=1}^t h_i^2
\]

\end{lemma}

Motivated by the above lemma, we present our new (equivalent) version of the SN density:

\begin{equation}\label{eq:sn_pdf}
    f(x|B) = \frac{1}{c(B)} \exp\left(-Z_x^\top B Z_x\right),
\end{equation}
where, $c(B)$ is the normalization constant. Specifically,
\begin{equation}\label{eq: c_B defn}
    c(B) := \int_X \exp \left(-Z_x^\top B Z_x\right) dx
\end{equation}
Observe that the top left element of $B$, say $B_{(0,0)}$ will contribute to a constant in the pdf, and will be absorbed in $c(B)$. 

The only difference between (\ref{eq:sn_pdf}) and (\ref{eq:sn_pdf_old}) is that we replace $\mu$ by an additional row and column in the $P$ matrix, to call it the $B$ matrix, and our $Z_x$ now has the additional constant term $1$.

To complete the argument, we have the following result:
\begin{theorem}\label{thm:equivalent}
    The two formulations of SN densities (\ref{eq:sn_pdf_old}) and (\ref{eq:sn_pdf}) are equivalent. In other words, any density function $f(x)$ that can be written as (\ref{eq:sn_pdf_old}) for some $(\mu,P)$ can also be expressed as (\ref{eq:sn_pdf}) for some PSD matrix $B$ and vice versa. 
\end{theorem}

 From the above theorem, we will use $B(\mu,P)$ to denote the parameter $B$ in the new formulation \ref{eq:sn_pdf_old} for a corresponding $(\mu,P)$ in the previous formulation. A direct corollary of the above is that:

\begin{corollary}
    The class of SN distributions is identical to the class of distributions of the form $f(x)\propto e^{-g(x)}$ where $g(x)$ is a SOS polynomial.
\end{corollary}

  Similar to the original SN setup, we use the MLE estimate:
\begin{align*}
    B^* &:= \arg\max_{B \succeq 0} \prod_{i=1}^n f(x_i \mid B) \\
        &= \arg\max_{B \succeq 0} \sum_{i=1}^n \log f(x_i \mid B) \\
        &= \arg\max_{B \succeq 0} \left[ -\sum_{i=1}^n Z_{x_i}^\top B Z_{x_i} - n \log c(B) \right] \\
        &= \arg\min_{B \succeq 0} \left[ \frac{1}{n} \sum_{i=1}^n Z_{x_i}^\top B Z_{x_i} + \log c(B) \right] \\
        &= \arg\min_{B \succeq 0} \left[ \frac{1}{n} \sum_{i=1}^n Z_{x_i}^\top B Z_{x_i} + \log \left( \int_{\mathcal{X}} \exp(-Z_x^\top B Z_x)\, dx \right) \right] \numberthis \label{eqn:convex_opt}
\end{align*}

  We next present the main result of this section:
  \begin{theorem}\label{thm: CVXOPT}
  The MLE search problem (\ref{eqn:convex_opt}) is a convex optimization problem
  \end{theorem}

We would like to emphasize that having a convex formulation does not directly provide an easy workaround to achieve the global optimum, particularly due to the constraint \( B \succeq 0 \). Although the problem is convex, we still rely on numerical optimization techniques. However, we can leverage well-established algorithms to tackle this issue effectively. Specifically, we utilize semidefinite programming (SDP) relaxations to approximate the global optimum.

Additionally, we can derive a feature-space MLE, \( \hat{B} \) obtained from the MLE of the previous formulation as:
\[
\hat{B} = B(\hat{\mu}, \hat{P}),
\]
as a direct consequence of Theorem~\ref{thm:equivalent}. This \( \hat{B} \) will serve as the starting point for the numerical search, facilitating the convergence to the global optimum.

\paragraph{Universal Approximation Motivation.}
We end this discussion by noting that the SN construction is motivated by the expressive power of polynomial-based log-densities on
compact domains. The key motivation behind proposing the SN distribution stems from
the idea that the multivariate distribution of the data being studied possesses a density function,
and we seek to model the log-density using high-dimensional polynomials. A useful justification for
this point of view is provided by the Stone--Weierstrass theorem.  

\begin{theorem}[Stone--Weierstrass Theorem]
Let \( X \) be any {compact space}, and let \( C(X) \) denote the set of all continuous functions \( f : X \to \mathbb{R} \), equipped with pointwise addition and multiplication. Let \( \mathcal{A} \) be a subalgebra of \( C(X) \). If \( \mathcal{A} \) contains the constant functions and separates points of \( X \) (i.e., for any two distinct points \( x, y \in X \), there exists \( f \in \mathcal{A} \) such that \( f(x) \ne f(y) \)), then \( \mathcal{A} \) is dense in \( C(X) \) under the uniform norm.
\end{theorem}
This is a generalization of the {Weierstrass approximation theorem}. In our context, the space \( C(X) \) contains the log-densities of random variables defined over a compact domain \( X \subset \mathbb{R}^p \). We approximate such log-densities using elements of a subalgebra \( \mathcal{A} \), which we take to be the set of multivariate polynomials. The class of polynomial is dense in \( C(X) \), justifying our approach. 

\subsection{Parameter Estimation for the Convex Reformulation}

For our problem, we define the objective function from \ref{eqn:convex_opt}:  
\begin{equation}
\label{eqn:objective}
J(B) := \frac{1}{n} \sum_{i=1}^n Z_{x_i}^\top B Z_{x_i}
+ \log \left( \int_{\mathcal{X}} \exp(-z_x^\top B z_x)\, dx \right).
\end{equation}
where \(Z_{x_i} \in \mathbb{R}^{m_z}\) are input feature vectors, \(B \in \mathbb{R}^{m_z \times m_z}\) is a PSD matrix, and \(c(B)\) represents the constant of proportionality. Our goal is to solve the following optimization problem:  

\begin{align*}
    \min_B &\quad J(B)\\
    s.t &\quad  B\succeq 0
\end{align*}

To address this constrained optimization problem, we employ the Frank-Wolfe algorithm \citep{Frank_wolfe}, also known as the conditional gradient method. This algorithm is particularly well-suited for convex optimization problems defined over compact convex domains. The Frank-Wolfe algorithm iteratively solves a linear approximation of the objective function to identify a feasible direction, followed by step-size optimization along that direction. Its computational efficiency and simplicity make it a popular choice for tackling large-scale problems with semidefinite constraints.

However, the constraint set \(B \succeq 0\) is not compact. To overcome this limitation, we artificially impose an additional bound on \(\|B\|\), resulting in the following feasible region:  
\[
\S := \{B \in \mathbb{R}^{m_z \times m_z} \mid B \succeq 0, \; \|B\| \leq M\},
\]  
where \(M > 0\) is a large constant. This bounding ensures that the feasible set \(S\) is compact, a necessary condition for guaranteeing convergence of the Frank-Wolfe algorithm.

This approach is conceptually similar to placing a lower bound on \(c(B)\), which is a common practice in maximum likelihood estimation (MLE) for the SN distribution. Recall that \(c(B)\) appears in the denominator of the SN probability density function. Extremely small values of \(c(B)\) can cause the density to blow up, leading to numerical instability during MLE search. To address this issue, it is customary to restrict \(c(B)\) to values above a certain threshold. By directly bounding \(B\), we achieve a similar effect: ensuring numerical stability while simultaneously making the feasible region compact, enabling the application of the Frank-Wolfe algorithm effectively. The Frank-Wolfe algorithm is summarized as follows:

\begin{enumerate}
    \item \textbf{Initialization}: Start with an initial feasible point \(B^{(0)}\), where \(B^{(0)}\in \S\).
    \item \textbf{Iterative Updates}: For \(t = 0, 1, 2, \dots\):
    \begin{enumerate}
        \item \textbf{Linearization}: Compute the gradient \(\nabla J(B^{(t)})\) and solve the linearized subproblem:  
        \begin{equation}\label{eqn: SDP}
            \tilde B^{(t)} := \arg \min_{B\in \S} \ \langle \nabla J(B^{(t)}), B \rangle.
        \end{equation}
       Procedure to solve the above problem is discussed in section \ref{sec: SDP_problem}

        \item \textbf{Update Direction}: Set the direction \(D^{(t)} = \tilde B^{(t)} - B^{(t)}\).
        \item \textbf{Step Size}: Choose a step size \(\gamma_t \in [0, 1]\). Generally taken as $\frac{2}{2+t}$.
        
        \item \textbf{Update}: Update \(B^{(t)}\) as:  
        \[
        B^{(t+1)} := B^{(t)} + \gamma_t D^{(t)}.
        \]
    \end{enumerate}
    \item \textbf{Stopping Criterion}: Use a suitable criterion to break the above loop. Some choices are discussed in \ref{sec: stopping_criteria}
\end{enumerate}

The Frank-Wolfe algorithm ensures that \(B^{(t)} \succeq 0\) at every iteration, maintaining feasibility with respect to the semidefinite constraint. Its key advantage lies in leveraging the linearized subproblem, which is computationally simpler than directly minimizing the original nonlinear objective.

\subsubsection{The SDP Problem} \label{sec: SDP_problem}

The above equation (\ref{eqn: SDP}) can be interpreted as to minimize the linear approximation of the problem given by the first-order Taylor approximation of $J(B)$ around $B^{(t)}$, constrained to stay within $\S$.  It can be shown that

\[
\nabla J\!\left(B^{(t)}\right)
=
\frac{1}{n} \sum_{i=1}^n Z_{x_i} Z_{x_i}^\top
\;-\;
\frac{
    \displaystyle\int_{\mathcal{X}} Z_x Z_x^\top \exp\left(-Z_x^\top B^{(t)} Z_x\right)\, dx
}{
    \displaystyle\int_{\mathcal{X}} \exp\left(-Z_x^\top B^{(t)} Z_x\right)\, dx
}.
\]

To approximate the latter integral, we fix some $K$ points in $\X$, namely $x_1,...,x_K$, and we use:
\[
\nabla \hat{J}\!\left(B^{(t)}\right)
=
\frac{1}{n} \sum_{i=1}^n Z_{x_i} Z_{x_i}^\top
\;-\;
\frac{
    \displaystyle\sum_{j=1}^K Z_{x_j} Z_{x_j}^\top \exp\left(-Z_{x_j}^\top B^{(t)} Z_{x_j}\right)
}{
    \displaystyle\sum_{j=1}^K \exp\left(-Z_{x_j}^\top B^{(t)} Z_{x_j}\right)
}.
\]

Thus, in each iteration, we are required to minimize $\langle \nabla \hat J(B^{(t)}),B\rangle$ (the dot product between the two quantities), subject to $B\in \S$ . This is a standard SDP problem (linear objective with PSD constraint) and can be solved using the CVXOPT package in python.

\subsubsection{Stopping Criteria}\label{sec: stopping_criteria}

The Frank-Wolfe algorithm can be terminated using various stopping criteria, depending on computational constraints and desired accuracy. Below, we outline commonly used conditions for stopping the iterations.

\begin{itemize}
    \item \textbf{Gradient Norm Threshold:} The algorithm halts when the norm of the gradient falls below a predefined threshold, indicating approximate stationarity:
    \begin{equation*}
        \|\nabla J(B^{(t)})\| \leq \epsilon.
    \end{equation*}
    \item \textbf{Objective Function Improvement:} The algorithm can terminate when the relative or absolute change in the objective function is small:
    \begin{equation*}
        |J(B^{(t)}) - J(B^{(t-1)})| \leq \delta.
    \end{equation*}
    This criterion ensures that iterations stop when further improvements become negligible.
    
    \item \textbf{Primal-Dual Gap:} We obtain a lower bound for the optimal objective function at each iteration. Since $J(B)$ is a convex function, for the optimal solution $B^*$ we have:
    \begin{align}
        J(B^*)&\geq J(B) + \langle\nabla J(B), B^*-B\rangle\\
        &\geq J(B) - \langle\nabla J(B), B\rangle + \min_{\tilde B\in \S}\langle\nabla J(B), \tilde B\rangle
    \end{align}
    Note that at every iteration, the later optimization problem is already solved in eqn \ref{eqn: SDP}. Let $l_t$ denote the lower bound at iteration $t$, with $l_0 = -\infty$. Thus, we update the lower bound as:
    $$l_t = \max\ \left(l_{t-1}\ ,\ J(B^{(t)}) + \langle\nabla J(B^{(t)}), \tilde B^{(t)}-B^{(t)}\rangle\right)$$
    where $\tilde B^{(t)}$ is obtained from eqn \ref{eqn: SDP}. At every iteration, we have that $$l_t\leq J(B^*) \leq J(B^{(t)})$$ We terminate when the duality gap $\left|J(B^{(t)}) - l_t\right| $ falls below a predefined threshold.

\end{itemize}

 While the primal-dual gap is often preferred in convex optimization, gradient-based criteria may be computationally inexpensive. In practice, a combination of these criteria is typically employed to balance computational efficiency and solution accuracy. A predefined maximum number of iterations is also enforced to ensure termination within a reasonable computational budget: $t \geq T_{\max}$.

\section{Approximate Procedures}\label{sec: approximate_procedures}
SN distributions provide a flexible framework for modeling multivariate data, capturing both multimodality and dependence structures simultaneously. However, estimating the parameters of SNs remains computationally demanding. To mitigate this challenge, several fast but suboptimal algorithms have been proposed, such as the feature-space maximum likelihood estimator (FMLE) and sliced exponential approximations. Importantly, these approaches can also be adapted to the convex reformulation developed in this paper. Beyond these existing methods, we introduce two new modifications designed to further accelerate the estimation of SNs.

\subsection{Existing Approaches}\label{sec: Existing_Approaches}

\begin{itemize}
    \item \textbf{Feature Space MLE (FMLE)}: Recall the feature-space maximum likelihood estimator (FMLE) introduced in Section~\ref{sec: SN parameter estimation}; see equations~\ref{FMLE_mu} and~\ref{FMLE_P}. We use the FMLE as an initialization for the numerical search for the true maximum likelihood estimator. Although the FMLE is not the exact maximizer of the likelihood, it is computationally convenient and straightforward to evaluate. It has been empirically shown in \cite{Crespo_SN_First} to often achieve a good likelihood in practice. Thus, when the dimension of the problem is very large, making the true likelihood search computationally intractable, the FMLE is often used as an initial condition.\par
The same idea may be implemented in the convex reformulation as well. According to  Theorem \ref{thm:equivalent}, we may obtain the FMLE: $\hat B$ for the convex reformulation (\ref{eq:sn_pdf}) corresponding to $(\hat \mu,\hat P)$. In particular, we set $\hat B$ as follows:
  \[\hat B
=\frac 12\begin{bmatrix}
\hat\mu^\top \hat P\hat \mu & -\hat\mu^\top \hat P \\
- \hat P \hat \mu & \hat P
\end{bmatrix}
\]

which will give a satisfactory ballpark estimate for the optimal likelihood in  \ref{eq:sn_pdf}.
    \item \textbf{Sliced Exponential}: \cite{sliced_exponential} introduces the concept of sliced exponential distributions where $f(x) \propto e^{-\lambda^\top Z_x}$. These are more flexible, as they are not restricted to S.O.S polynomials, however may lead to spurious distributions. Thus, a carefull restriction of the search space is proposed, called the primal SNs \cite{Hammond2024-cn}. They obtain a SOS basis using the polynomials that occur from the FMLE solution, say $Z_{sos}$, and consider   $f(x) \propto e^{-\lambda^\top Z_{sos}}$. This approach remains identical in the convex reformulation as well.

    \item \textbf{Dimension Reduction}: In the next section,  we consider the independence grouping approach introduced in \cite{dim_red}
. Let $\{x_1,...,x_m\}$ be the complete set of variables we are interested in modeling. Let $G_i$ denote disjoint subgroups of indices: $G_i = \{{g_{i1}},...,{g_{im_i}}\} $, with $\cup_i G_i = \{1,...,m\}$. Let $x_{G_i}$ denote the set of variables: $x_{G_i}:=\{x_j: j\in G_i\}$. Thus, any collection of disjoint $G_i's$ provide a partition for the set of variables in $\{x_1,...,x_m\}$. Suppose it was possible to obtain such a partition of size $I$, for which each of the $x_{G_i}$ are independent among themselves. In that case, the joint PDF of $x$ can be written as the product of the PDFs of the subgroups:
\begin{equation}\label{sec: dimension reduction}
    f(x_1,...,x_m) = \prod^I_{i=1} f_{G_i}(x_{G_i})
\end{equation}

As a consequence, the $P$ matrix for the corresponding SN distribution can be obtained as follows:
  \[P
=\begin{bmatrix}
P_1 & 0 & ... & 0 & 0 \\
0 & P_2 & ... & 0 & 0 \\
... &...&...&...&... \\
0 & 0 & ... & P_I & 0 \\
0 & 0 & ... & 0 & 0
\end{bmatrix}
\]
 where each of the subgroups $x_{G_i}$ follows a SN distribution with parameters $(\mu_i,P_i)$. The last rows (and columns) of zeros corresponds to the cross terms between these subgroups. 
 
 However, these subgroups are not readily available and require the identification of independent random variables. \cite{dim_red} provide a framework to address this problem by introducing a notion of dependency, denoted as \( \rho(x_i, x_j; \mathcal{D}) \), which quantifies the dependence between two variables \( x_i \) and \( x_j \) given the dataset \( \mathcal{D} \). Their approach involves comparing the empirical bivariate copula against the theoretical copula under the assumption of complete independence. Two variables are deemed independent if \( \rho(x_i, x_j; \mathcal{D}) < \epsilon \) for some predefined threshold \( \epsilon \).

\end{itemize}

\subsection{Leveraging Weak Dependence }\label{sec:leveraging-low-dependence}

In this section we expand upon the dimension reduction approach of \cite{dim_red}. 
For a given dataset, completely independent subgroups of variables may not exist, 
since exact independence is largely a theoretical ideal. We therefore introduce an 
empirical modification of the original approach for settings where independence does 
not strictly hold but the dependence among subgroups is sufficiently weak. 

The degree of dependence between variables (or groups of variables) can be quantified 
in several formal ways. In this work, we adopt the \emph{distance correlation} 
$\rho_D(\cdot,\cdot)$, which captures both linear and 
nonlinear dependence. It is a widely used measure of dependence among random variables introduced by \cite{dist_corr}. Unlike Pearson correlation, distance correlation has the crucial property that \(\rho_D = 0\) if and only if the variables are truly independent. Moreover, it provides a statistical framework for formally testing independence, making it a robust choice for dependency assessment. Given these advantages, we adopt distance correlation (\(\rho_D\)) as our measure of independence and apply the same hierarchical clustering methodology as in \cite{dim_red} to estimate the subgroup structure.

Let $X=(X_1,\dots,X_p)^\top \in \mathbb{R}^p$ denote the feature vector. The proposed 
procedure consists of the following steps.

\begin{enumerate}

\item \textbf{Step 1 (Grouping variables via weak dependence).}  
We partition the variables into groups such that variables within a group exhibit 
pairwise dependence, however, is weakly dependent or independent of variables outside the group. The grouping is constructed using a greedy sequential procedure. Let $\epsilon>0$ be a pre-specified threshold. Initialize with the first variable:
\[
G_1 = \{1\}, \qquad \mathcal{G} = \{G_1\},
\]
where $\mathcal{G}$ denotes the collection of groups. Then, for $j=2,\dots,p$, process 
$X_j$ as follows:

\begin{itemize}
    \item Compute $\rho_D(X_i, X_j)$ for all $i<j$.
    \item If there exists a group $G \in \mathcal{G}$ and some $i \in G$ such that
    \[
    \rho_D(X_i, X_j) > \epsilon,
    \]
    assign $j$ to that group, i.e., update $G \leftarrow G \cup \{j\}$.
    \item Otherwise, create a new group $G_{\text{new}} = \{j\}$ and update
    $\mathcal{G} \leftarrow \mathcal{G} \cup \{G_{\text{new}}\}$.
\end{itemize}

At the end of this step, we obtain a partition
\[
\{1,\dots,p\} = G_1 \cup \cdots \cup G_K, \qquad G_g \cap G_{g'} = \varnothing \ (g\neq g'),
\]
where each group $G_g$ contains variables that are approximately independent of 
variables in other groups in the sense of distance correlation.

\item \textbf{Step 2 (Within-group modeling).}  
For each group $G_g$, let $X_{G_g}$ denote the subvector of $X$ restricted to the 
indices in $G_g$. We model the distribution of $X_{G_g}$ using a SN distribution with parameters $(\mu_g, P_g)$, where $\mu_g \in \mathbb{R}^{|G_g|}$ is 
the mean parameter and $P_g \in \mathbb{R}^{|G_g|\times|G_g|}$ is the precision matrix 
. These parameters are estimated independently for each group.

This step yields a block-diagonal precision structure
\[
P_{\text{block}} = \operatorname{diag}(P_1,\dots,P_K),
\]
which corresponds to an approximation assuming group-wise independence.

\item \textbf{Step 3 (Introducing cross-group dependence).}  
To account for residual dependence across groups, we relax the block-diagonal structure 
by introducing off-diagonal precision blocks. For simplicity, consider two groups 
($K=2$) with precision blocks $P_{11}=P_1$ and $P_{22}=P_2$. We complete the full 
precision matrix as
\[
P = 
\begin{pmatrix}
P_{11} & P_{12} \\
P_{12}^\top & P_{22}
\end{pmatrix},
\]
where $P_{12}$ captures cross-group interactions. During this step, $P_{11}$ and $P_{22}$ 
are kept fixed, and we optimize over $P_{12}$ subject to the constraint
\[
P \succeq 0.
\]
For $K>2$, the same idea applies by introducing off-diagonal blocks between all pairs 
of groups. This stage allows the model to capture weak inter-group dependence while 
preserving the stability gained from the initial dimension reduction.

Section 5.1 of \cite{dim_red} provides a sequential approach in estimating the full $P$, by fixing the $P_{11}$ and $P_{22}$. Using Schur complement properties, $P\succ 0$ if and only if:
\begin{equation}\label{eq schur}
    P_{11} - P_{12} P_{
22}^{-1}P_{21} \succ 0
\end{equation}

We reparameterize $P_{12}$ :
\[
P_{12}(U) = P_{1}^{1/2} \, C(U) \, P_{2}^{1/2}, \quad \text{where} \quad C(U) := \frac{U}{1 + \|U\|_F},
\]
for an unconstrained matrix variable $U \in \mathbb{R}^{d_1 \times d_2}$. By this, we obtain:
\[
P_{1} - P_{12} P_{2}^{-1} P_{12}^\top 
= P_{1}^{1/2} \left( I - C(U) C(U)^\top \right) P_{1}^{1/2}.
\]
Thus, the Schur complement is positive definite if and only if $\|C(U)\|_2 < 1$, which is guaranteed by the definition of $C(U)$.
This maps any $U \in \mathbb{R}^{d_1 \times d_2}$ to a matrix with $\|C(U)\|_2 \leq \|C(U)\|_F < 1$, ensuring $P \succ 0$ without explicit constraints.

\paragraph{What is optimized:}
Letting $Z_1, Z_2$ denote the feature matrices for the two blocks, and $Z_1^{(j)}, Z_2^{(j)}$ the integration grid, we define the full precision matrix:
\[
P(U) = 
\begin{bmatrix}
P_1 & P_{12}(U) \\
P_{12}(U)^\top & P_2
\end{bmatrix}, 
\]
Since $\mu_i$ are kept fixed, define centered blocks $y_1:=z_1-\mu_1$ and $y_2:=z_2-\mu_2$. Then
\begin{align}
\begin{pmatrix}y_1\\y_2\end{pmatrix}^{\!\top}
P
\begin{pmatrix}y_1\\y_2\end{pmatrix}
&=
\begin{pmatrix}y_1\\y_2\end{pmatrix}^{\!\top}
\begin{pmatrix}
P_{11} & P_{12}\\
P_{12}^\top & P_{22}
\end{pmatrix}
\begin{pmatrix}y_1\\y_2\end{pmatrix} \notag\\
&= y_1^\top P_{11}y_1 \;+\; y_2^\top P_{22}y_2 \;+\; 2\,y_1^\top P_{12}y_2.
\label{eq:quad_block_expansion}
\end{align}
For data points $x^{(i)}$, $i = 1,...,n$, and corresponding $y^{(i)}$, the SN likelihood function can be written as:
\begin{align*}
\ell(P) &= \prod_{i=1}^n\frac{1}{c(P)}\exp{\left\{\frac 12\left(-{y^{(i)}_1}^\top P_{11}y^{(i)}_1 \;-\; {y^{(i)}_2}^\top P_{22}y_2 \;-\; 2\,{y^{(i)}_1}^\top P_{12}y^{(i)}_2\right)\right\}}\\
&= K\cdot \frac{\prod_{i=1}^n \exp{\left(-\,{y^{(i)}_1}^\top P_{12}y^{(i)}_2\right)}}{c(P)^n} 
\end{align*}

where $K$ is a constant depending on $P_{11}, P_{22}$ and $n$. Now, maximizing the likelihood is equivalent to minimizing the following:
\begin{equation}\label{eqn: P_12_obj}
    -\frac{1}n\log(\ell(P)) = -\frac1n\log(K)+\frac{1}{n} \sum_{i=1}^n y_1^{(i)\top} P_{12} y_2^{(i)}+\log(c(P))
\end{equation}

Now,

\begin{align*}
    c(P) &= \int_{\X} \exp{\left(-\frac 12{y_1}^\top P_{11}y_1 \;-\; \frac 12{y_2}^\top P_{22}y_2 \;-\; \,{y_1}^\top P_{12}y_2\right)}\\
    &= \int_{\X} \exp{\left(-\frac 12{y_1}^\top P_{11}y_1 \;-\; \frac12{y_2}^\top P_{22}y_2 \;\right)}\cdot  \exp{\left(-\,{y_1}^\top P_{12}y_2\;\right)}
\end{align*}

Now, we approximate the above integral using a summation. Suppose $x^{(j)}$ are points to form a uniform grid on $X$, $j = 1,...,m$. Correspondingly, we also have $y^{(j)}:=x^{(j)}-\mu$. Define constants $w^{(j)}$ depending on $P_{11}$ and $P_{22}$:
$$w^{(j)}:= \exp{\left(-\frac 12{y^{(j)}_1}^\top P_{11}y^{(j)}_1 \;-\; \frac 12{y^{(j)}_2}^\top P_{22}y^{(j)}_2 \;\right)}$$
Thus, we have:
\begin{equation}
\label{eq:cP_grid_approx_weights}
c(P)\;\approx\;\frac{\mathrm{Vol}({X})}{m} \sum_{j=1}^m w^{(j)} \exp\!\Big(-{y_1^{(j)}}^\top P_{12}y_2^{(j)}\Big).
\end{equation}

 The only variable here is the matrix $P_{12}$, which we parameterized with the matrix $U$. Thus, substituting the above approximation (\ref{eq:cP_grid_approx_weights}) into the objective (\ref{eqn: P_12_obj}), and ignoring the constants, we have:
 \begin{equation}
     \min_{U \in \mathbb{R}^{d_1 \times d_2}} \;
\frac{1}{n} \sum_{i=1}^n y_1^{(i)\top} P_{12}(U) y_2^{(i)} +
\log \left( \frac{1}{m} \sum_{j=1}^m w^{(j)}\exp\left(-y_1^{(j)\top} P_{12}(U) y_2^{(j)} \right) \right)
 \end{equation}

This objective is smooth and unconstrained over $U$, and the optimal $P_{12}$ is recovered via the transformation $P_{12}(U) = P_{1}^{1/2} \, C(U) \, P_{2}^{1/2}$. The full joint precision matrix $P$ is assembled from the known $P_1$, $P_2$, and the optimized $P_{12}$.
\end{enumerate}

\subsubsection{Limitation: Marginals will be compromised}
\label{subsec:marginal-drift}

In our three–stage weak–dependence procedure, we first partition the features into groups. Then, we estimate separate SN models on variable groups $G_1$ and $G_2$, obtaining $P_{11}$ and $P_{22}$ that (regularized-)maximize the marginal likelihoods within their respective SN classes. These subgroup fits are therefore already marginal–optimal given the chosen feature maps. In the third stage, we introduce cross–group dependence by completing the precision with an off–diagonal block $P_{12}$ while keeping $P_{11}$ and $P_{22}$ fixed and enforcing $P\succ0$. This completion unavoidably alters the induced marginals. Writing the joint SN density (up to normalization) as
\[
p(x_1,x_2)\ \propto\ 
\exp\!\Big(-z_1(x_1)^\top P_{11} z_1(x_1) - z_2(x_2)^\top P_{22} z_2(x_2) - 2\, z_1(x_1)^\top P_{12} z_2(x_2)\Big),
\]
the new marginal on $x_1$ becomes
\begin{align*}
    &\  p_1^{\text{new}}(x_1)\ \\
    \propto\ &\ \exp\!\big(-z_1(x_1)^\top P_{11} z_1(x_1)\big)\,
\underbrace{\int \exp\!\big(- z_2(x_2)^\top P_{22} z_2(x_2) - 2\, z_1(x_1)^\top P_{12} z_2(x_2)\big)\,dx_2}_{g(x_1;P_{12})}
\end{align*}

Unless $P_{12}=0$, the factor $g(x_1;P_{12})$, which depends on $x_1$ through the bilinear coupling $z_1^\top P_{12} z_2$, so $p_1^{\text{new}}$, is no longer proportional to the stage–1 marginal $\exp(-z_1^\top P_{11} z_1)$. An identical conclusion holds for the marginal on $x_2$. Consequently, because $P_{11}$ and $P_{22}$ were already chosen to optimize their subgroup objectives, introducing any nonzero $P_{12}$ cannot improve those marginals and typically makes them suboptimal relative to their stage–2 optima. The matrix–completion step should therefore be viewed as an explicit trade–off: we gain joint dependence fit via $P_{12}$ at the expense of a degradation in the marginal fit.

\section{Application}\label{sec: numerical_example}

We apply the proposed SN modeling pipeline to the NASA Langley flight dataset, which was collected in experiments motivated by in-flight loss of control (LOC). LOC constitutes the largest fatal accident category for commercial jet airplanes worldwide \citep{Belcastro11}. In general terms, LOC describes aircraft motion that departs from the normal operating flight envelope and is not predictably corrected by pilot inputs. Such behavior is often associated with strong nonlinear effects and dynamic coupling \citep{Belcastro11,Crespo12}. In these regimes, small changes in the vehicle state may lead to disproportionately large responses, including oscillatory or divergent motion. The resulting uncommanded angular rates can make it difficult to maintain heading, altitude, and wings-level flight.

The data arise from experiments conducted with the Generic Transport Model (GTM), a \(5.5\%\) dynamically scaled, remotely piloted, twin-turbine aircraft. The flights used in our analysis correspond to critical upset conditions. While these runs are nominally identical, their recorded responses show substantial variability, consistent with the complex and nonlinear dynamics present in LOC-related flight behavior. This variability makes the dataset a useful testbed for the SN framework, which is designed to represent nonlinear dependence structures in a compact and tractable form.

We examine two numerical settings. In the first one, we fit an SN model to a three dimensional slice of the dataset, allowing us to evaluate the model's ability to capture localized nonlinear dependence in a low-dimensional setting. In the second one, we consider a five-dimensional subset of the data and use the weak-dependence assembly strategy from Section~\ref{sec:leveraging-low-dependence} with a \(2{+}3\) block partition. This setting tests whether separately fitted low-dimensional SN models can be assembled into a higher-dimensional representation through cross-block completion.


\paragraph{Three-dimensional dataset.}
We first consider the three-dimensional slice using \(n=1000\) observations and
polynomial degree \(d=2\). Figure~\ref{fig:nasa_3d_compare} compares the original
sample with MCMC draws from two fitted SN models: the feature-space maximum
likelihood estimator (FMLE) and the numerically optimized SN estimator. Both fitted
models reproduce the main curved and clustered structures present in the observed
data, demonstrating the ability of the polynomial feature representation to capture
nonlinear dependence. The FMLE provides a computationally inexpensive initial fit,
whereas the subsequent likelihood optimization further refines the estimated density.

Evaluated on the observed three-dimensional dataset, the FMLE and optimized SN
models attain log-likelihoods of \(-3.7
    \ \text{and}\ 
    \num{2.4},
\)
respectively, indicating that the numerical optimization improves
the in-sample fit relative to the FMLE initialization.

\begin{figure}[H]
    \centering

    \begin{subfigure}[t]{0.32\textwidth}
        \centering
        \includegraphics[width=\linewidth]{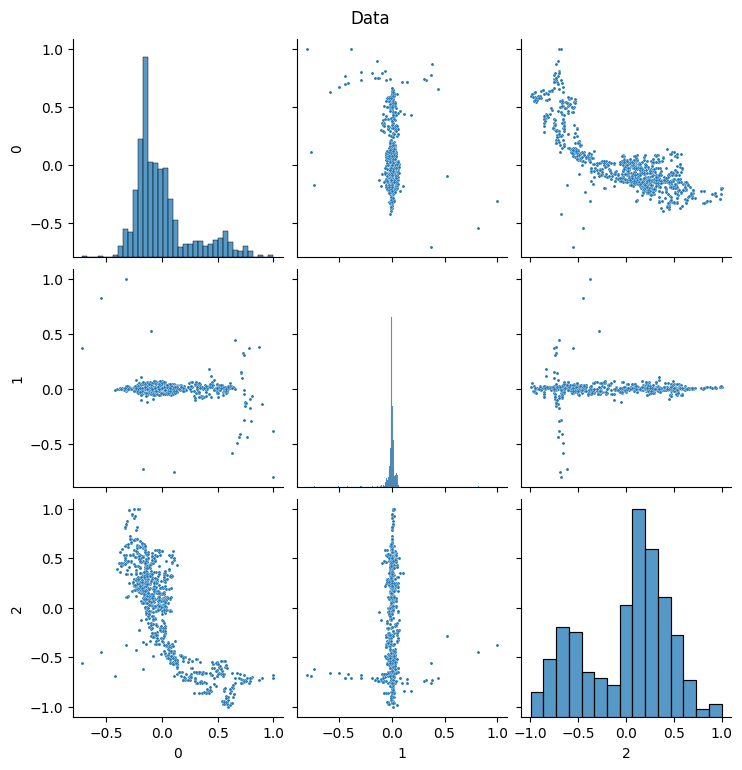}
        \caption{Original sample.}
        \label{fig:nasa_012_org}
    \end{subfigure}
    \hfill
    \begin{subfigure}[t]{0.32\textwidth}
        \centering
        \includegraphics[width=\linewidth]{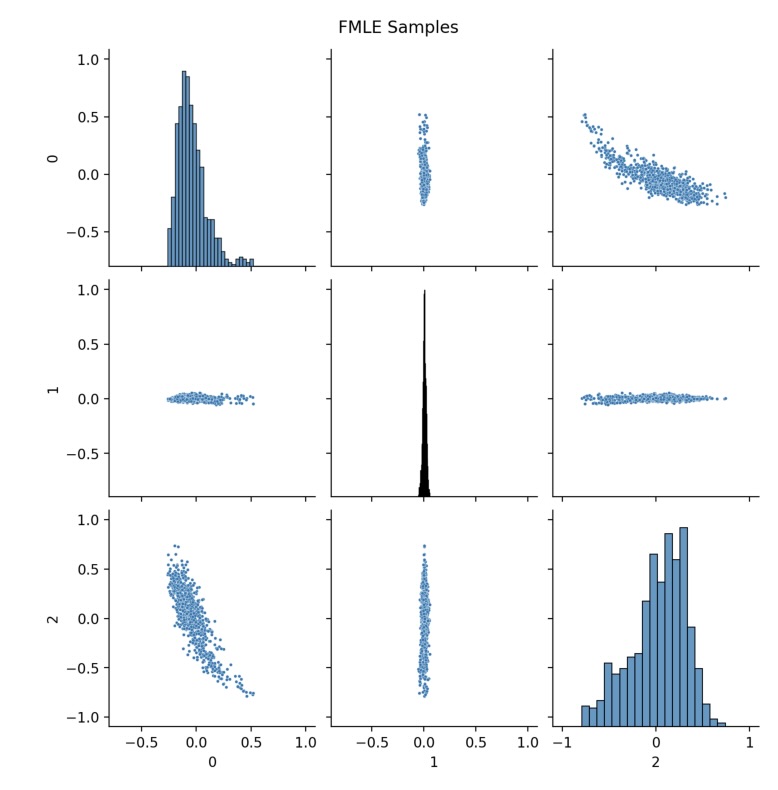}
        \caption{FMLE MCMC samples.}
        \label{fig:nasa_012_fmle}
    \end{subfigure}
    \hfill
    \begin{subfigure}[t]{0.32\textwidth}
        \centering
        \includegraphics[width=\linewidth]{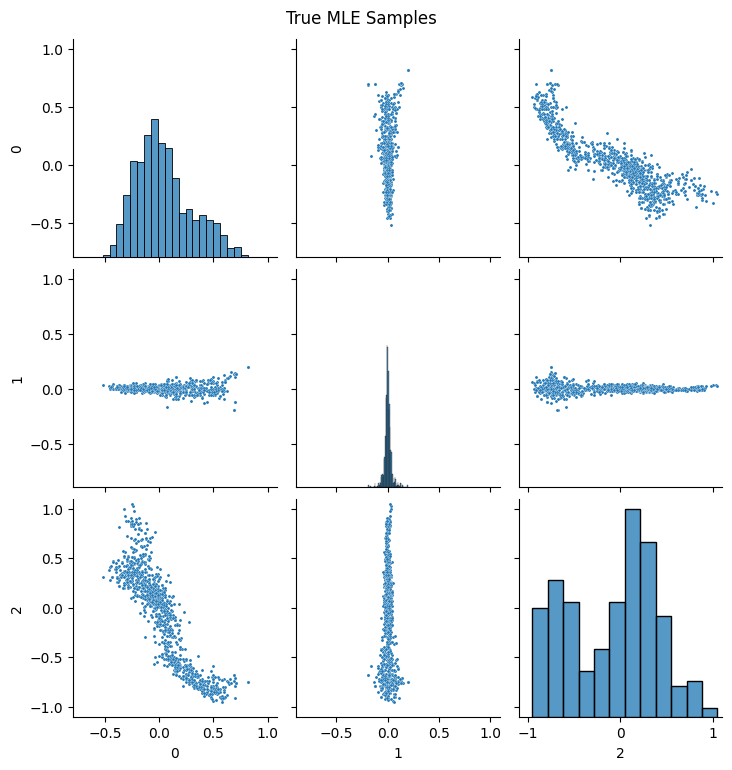}
        \caption{Optimized SN MCMC samples.}
        \label{fig:nasa_012_sn}
    \end{subfigure}

    \caption{Three-dimensional dataset from the LOC flight data. The original
    observations are compared with MCMC samples generated from the FMLE and the
    numerically optimized SN model. Both fitted models capture the principal nonlinear
    dependence patterns, while the likelihood comparison quantifies the improvement
    obtained by optimizing beyond the FMLE initialization.}
    \label{fig:nasa_3d_compare}
\end{figure}

\paragraph{Five-dimensional dataset with a \(2{+}3\) block decomposition.}
We next consider a five-dimensional dataset, again with \(n=1000\) observations, and use polynomial degree \(d=3\). Following Step 1 of Section~\ref{sec:leveraging-low-dependence}, the variables are partitioned into a two-dimensional block \(G_1\) (columns 0 and 4) and a three-dimensional block \(G_2\) (columns 1, 2 and 3). We then carry out the remaining stages of the weak-dependence assembly procedure: first, we fit subgroup SN models to estimate \(P_{11}\) on \(G_1\) and \(P_{22}\) on \(G_2\); second, we optimize the off-diagonal block \(P_{12}\) using the Schur-safe parameterization while holding \(P_{11}\) and \(P_{22}\) fixed. This yields a full precision matrix assembled from subgroup fits together with a cross-block completion step.

Figure~\ref{fig:nasa_5d_indep} shows samples from the assembled model under the independence approximation \(P_{12}=0\), while Figure~\ref{fig:nasa_5d_complete} shows samples after optimizing \(P_{12}\). The independence model preserves the internal structure of each subgroup but does not capture cross-block interactions. After optimizing \(P_{12}\), the completed model recovers salient cross-group dependence patterns while maintaining positive definiteness by construction. Table~\ref{tab:nasa_likelihood_summary} summarizes the log-likelihoods of the two assembled five-dimensional models.

Consistent with the discussion in Section~\ref{subsec:marginal-drift}, activating \(P_{12}\) can slightly deform subgroup marginals relative to their stage-1 optima. This reflects an inherent trade-off between preserving subgroup-optimal marginal fits and improving the quality of the overall joint fit.

\begin{figure}[H]
    \centering
    \includegraphics[width=1\textwidth]{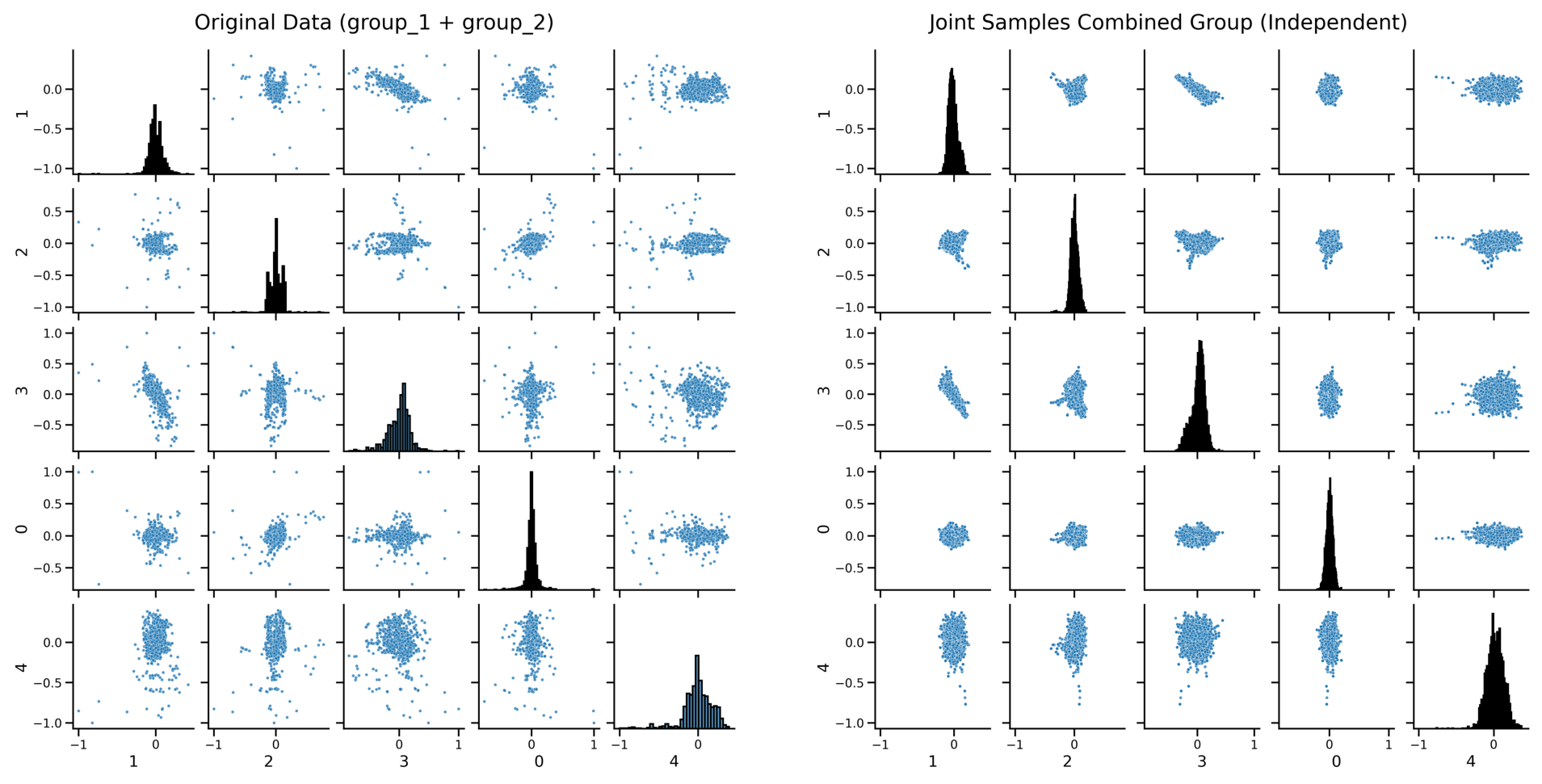}
    \caption{Five-dimensional dataset: true data on the left and samples from the assembled SN model under the independence approximation (\(P_{12}=0\)) on the right.}
    \label{fig:nasa_5d_indep}
\end{figure}

\begin{figure}[H]
    \centering
    \includegraphics[width=1\textwidth]{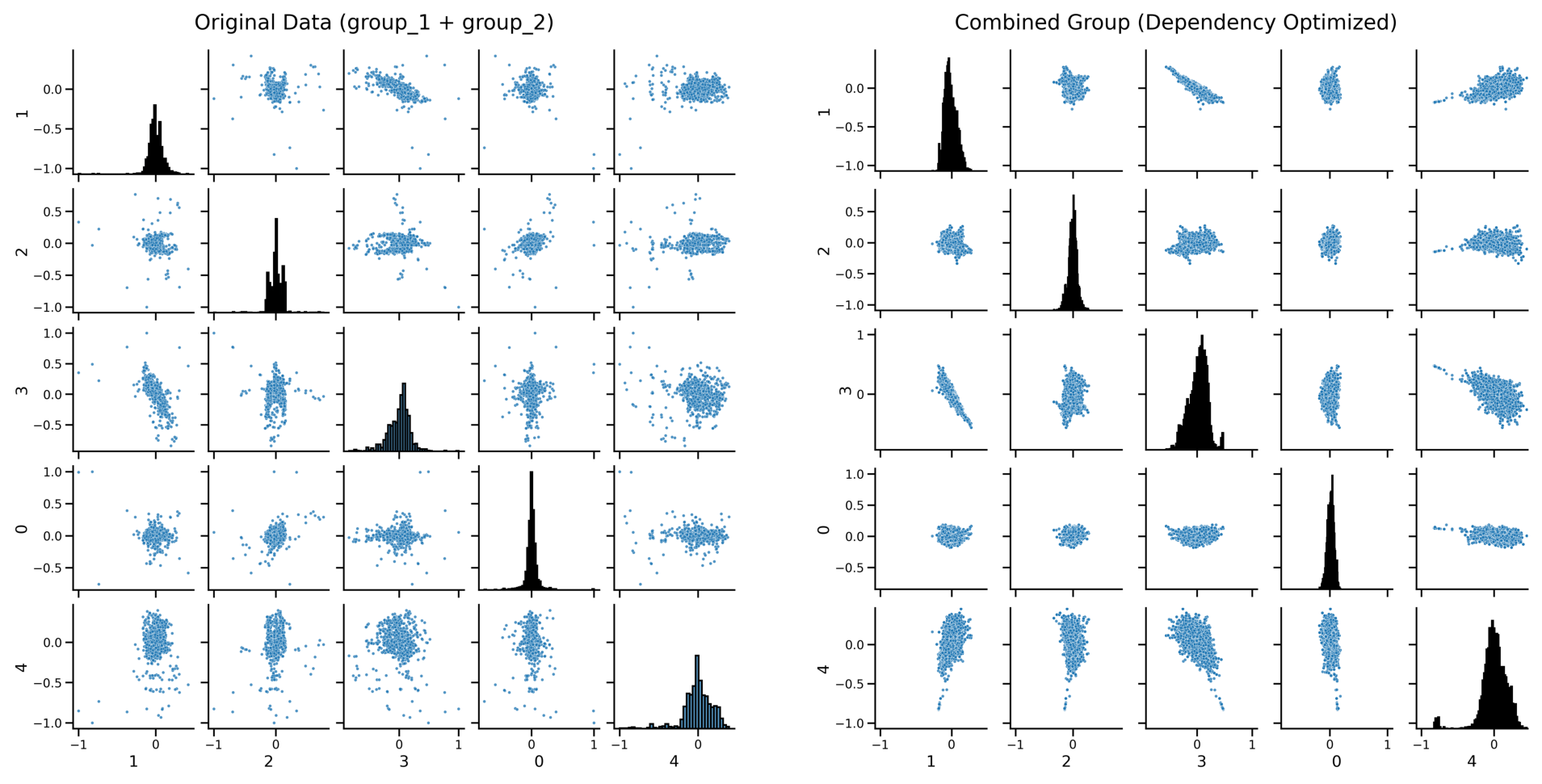}
    \caption{Five-dimensional dataset: assembled SN model after cross-block completion (\(P_{12}\neq 0\)).}
    \label{fig:nasa_5d_complete}
\end{figure}

\begin{table}[H]
    \centering
    \caption{Log-likelihood summary for the NASA numerical example.}
    \label{tab:nasa_likelihood_summary}
    \renewcommand{\arraystretch}{1.15}
\begin{tabular}{@{}>{\raggedright\arraybackslash}p{7.2cm}c@{}}
    \toprule
    \textbf{Model / Dataset} & \textbf{Log-likelihood} \\
    \midrule
    5D assembled SN (independence, \(P_{12}=0\))
        & {\num{1.1e3}} \\
    5D assembled SN (optimized \(P_{12}\))
        & {\num{1.4e3}} \\
    \bottomrule
\end{tabular}
\end{table}

\section{Limitation: Identifiability}\label{sec: limitaion}
Model identifiability refers to the ability to uniquely determine the model's parameters based on the observed data and the specified functional form of the model. A model is said to be identifiable if different parameter values produce distinct distributions of the observed data, ensuring that the parameters can be uniquely recovered given sufficient data. A typical reason for non-identifiability is over-parametrization of the model. The SN model suffers from this deficiency. To demonstrate,  consider the quadratic form \( Z_x^T B Z_x \), where \( Z_x = [1, x, x^2, x^3]^T \) and \( B \) is a symmetric matrix with entries \( b_{ij} \). The quadratic form can be expressed as:
\[
Z_x^T B Z_x = 
\begin{bmatrix}
1 & x & x^2 & x^3
\end{bmatrix}
\begin{bmatrix}
b_{11} & b_{12} & b_{13} & 1 \\
b_{21} & b_{22} & 1 & b_{24} \\
b_{31} & 1 & b_{33} & b_{34} \\
1 & b_{42} & b_{43} & b_{44}
\end{bmatrix}
\begin{bmatrix}
1 \\ x \\ x^2 \\ x^3
\end{bmatrix}.
\]
Alternatively, one can reparameterize \( B \) as:
\[
\begin{bmatrix}
b_{11} & b_{12} & b_{13} & 0 \\
b_{21} & b_{22} & 2 & b_{24} \\
b_{31} & 2 & b_{33} & b_{34} \\
0 & b_{42} & b_{43} & b_{44}
\end{bmatrix}.
\]

Despite these two different representations of \( B \), the resulting coefficient of \( x^3 \) in \( Z_x^T B Z_x \) remains identical. This demonstrates that different parametrizations of \( B \) can produce the same SN.  As a result, the optimization of such models lacks uniqueness, and the optima are not guaranteed to be unique. We would like to note here that this is not typical of the convex reformulation. The SN distribution in the original form is also over parameterized, and different choices of $P$ in (\ref{eq:sn_pdf_old}) may lead to the same distribution. The following developments yield a unique representation of the SN distribution.

\textbf{Notation: }For ease of notation, we will use $m$-dimensional vectors to refer to the rows (and columns) of $B$. Let $\balpha = (\alpha_1, \ldots, \alpha_m)$ denote the index corresponding to the monomial $x_1^{\alpha_1} \cdots x_m^{\alpha_m}$, where $\alpha_i \in \mathbb{N}$ and $1 \leq \sum_{i=1}^m \alpha_i \leq d$. 

Note that $Z_x^T B Z_x$ is a polynomial of degree $2d$ in $m$ variables. Thus, it can be uniquely specified by its $\binom{2d+m}{m}$ coefficients:
\begin{equation*}
    Z_x^T B Z_x = \theta_0 + \sum_{1 \leq \sum_{i=1}^m \beta_i \leq 2d} \theta_{\bbeta} \cdot x_1^{\beta_1} \cdots x_m^{\beta_m},
\end{equation*}
where $\bbeta = (\beta_1, \ldots, \beta_m)$ is the multi-index for the exponents of the monomials, and $\theta_{\beta_1 \cdots \beta_m}$ denotes the unique coefficients of the polynomial $Z_x^T B Z_x$. These coefficients can be deduced for a given matrix $B$. For a given $\bbeta$, consider the set of index pairs:
\begin{equation*}
    I_{\bbeta} = \left\{ (\balpha^{(1)}, \balpha^{(2)}) \;\middle|\; 
    \alpha^{(1)}_1 + \alpha^{(2)}_1 = \beta_1, \; \ldots, \; \alpha^{(1)}_m + \alpha^{(2)}_m = \beta_m \right\}.
\end{equation*}

Finally, we have:
$$\theta_{\bbeta} = \sum_{(\balpha^{(1)}, \balpha^{(2)})\in I_{\bbeta}} B_{\balpha^{(1)}, \balpha^{(2)}}$$
We may also fix $\theta_0$ to be 1, as it will be absorbed in the proportionality constant. Consider the $D = \binom{2d+m}{m}-1$ sized vector:
$$\Theta = \begin{pmatrix}
    \theta_{\bbeta^{(1)}},\ \cdots\ \theta_{\bbeta^{(D)}}
\end{pmatrix}$$
Thus, we may consider this $\Theta$ as the \textit{true} parameter of the problem, and the matrix $B$ as our auxiliary parameter. $\theta$ can be uniquely determined for any given $B$, but multiple $Bs$ can lead to the same $\Theta$. Two models may be considered identical if the $\Theta$'s are equal.

However, this reparameterization does not help to improve the likelihood of the model. On the other hand, adding this as a restriction during estimation will further complicate the model search, making it computationally difficult to arrive at a optima. Thus, we will formulate the problem in terms of the matrix $B$. However, insignificant changes in theta can trigger the termination of the search.

\section{Future Directions}

We conclude by outlining several open problems that arise naturally from the present work and whose resolution would significantly enhance both the theoretical understanding and practical applicability of Sliced Normal (SN) models.

\paragraph{1. Invertible transformations to a standard normal.}
A fundamental question is whether there exists an explicit, invertible transformation
\[
T:\mathbb{R}^p \to \mathbb{R}^p
\]
such that, for $X\sim \text{SN}(\mu,P)$,
\[
Z := T(X) \sim \mathcal{N}(0,I_p).
\]
This would play a role analogous to normalizing flows, but with the transformation derived from the analytical structure of SN distributions rather than from a learned neural network parameterization. Constructing such a map would have multiple benefits. This would enable efficient sampling from SN distributions without resorting to MCMC methods, which can be computationally expensive and sensitive to tuning in high dimensions. A direct transformation-based sampler would be particularly valuable in rare-event estimation and Reliability-Based Design Optimization (RBDO), where repeated sampling from complex distributions is required.

\paragraph{2. Model complexity, monomial growth, and adaptive degree selection.}
A second open problem concerns the growth of model complexity as the polynomial feature degree increases. The number of monomials grows rapidly with both the ambient dimension $p$ and the degree $d$, which limits the range of $(p,d)$ for which reliable estimation and optimization are feasible in practice.

This motivates the construction of a sequence of SN models of increasing degree,
\[
\text{SN}_1 \subset \text{SN}_2 \subset \cdots,
\]
in which inconsequential monomials are systematically identified and removed at each stage. Such a procedure would amount to an adaptive degree-selection or feature-pruning mechanism within the SN framework. A principled solution would allow high-degree SNs to be used in practice while keeping the effective parameter dimension under control, thereby extending the applicability of the framework to higher-dimensional and more complex data settings.


\section*{Acknowledgements}

The authors gratefully acknowledge the support of NASA's Human Research
Program (HRP), which provided the funding for this work.


\bibliographystyle{unsrtnat}
\bibliography{bibliography}


\section{Appendix}

We will use the following well known result from matrix algebra:
\begin{lemma}\label{lemm:schur_comp} Given any symmetric matrix, 
\[
M = \begin{bmatrix} A & B \\ B^\top & C \end{bmatrix},
\]
the following conditions are equivalent:
\begin{enumerate}
    \item \( M \succeq 0 \) (i.e., \( M \) is positive semidefinite).
    \item \( A \succeq 0, \, (I - AA^\dagger)B = 0, \, C - B^\top A^\dagger B \succeq 0 \).
    \item \( C \succeq 0, \, (I - CC^\dagger)B^\top = 0, \, A - BC^\dagger B^\top \succeq 0 \).
\end{enumerate}
\end{lemma}
\subsection{Proofs of Section \ref{sec: Convex_reformulation}}
\begin{proof}{Proof of Theorem \ref{thm:equivalent}}
Given any $\mu,P$ from the sliced normal representation \ref{eq:sn_pdf_old}, we set $B$ as follows:
  \[B(\mu, P)
=\frac 12\begin{bmatrix}
\mu^\top P\mu & -\mu^\top P \\
- P \mu & P
\end{bmatrix}
\]

Suppose:
\[
Z_x = \begin{bmatrix}
    1\\
    \tilde Z_x
\end{bmatrix}.
\]
Observe that:
\begin{align*}
    Z_x^\top BZ_x & =  \frac 12\begin{bmatrix}
    1 &
    \tilde Z_x^\top
\end{bmatrix}\begin{bmatrix}
\mu^\top P\mu & -\mu^\top P \\
- P \mu & P
\end{bmatrix}  \begin{bmatrix}
    1\\
    \tilde Z_x
\end{bmatrix}\\[0.3cm]
& = \frac 12 \left[\tilde Z_x^\top P \tilde Z_x - 2 \tilde Z_x^\top P \mu +\mu^\top P\mu\right] \\[0.3cm]
& = \frac 12 (\tilde Z_x-\mu)^\top P (\tilde Z_x-\mu)
\end{align*}

Hence, we have:
\begin{align*}
f(x| (\mu,P)) &= \frac{1}{c(\mu,P)} \exp\left(-\frac 12 (\tilde Z_x-\mu)^\top P (\tilde Z_x-\mu) \right)\\
    & =  \frac{1}{c(\mu,P)} \exp\left(- Z_x^\top BZ_x \right)
\end{align*}
Thus, with this $B$, the new representation \ref{eq:sn_pdf}  gives the same model as the sliced normal in \ref{eq:sn_pdf_old}. 
Also, by construction, $P$ is PSD $\implies B$ is PSD.
On the contrary given any $B$, we partition $B$ as follows:
 \[B
=\begin{bmatrix}
a_{1\times 1} &b^\top_{1\times m_z} \\
b_{m_z\times 1} & D_{m_z\times m_z}
\end{bmatrix}
\]
we may recover $P$ from the bottom right $m_z\times m_z$ sub-matrix of $B$, i.e , we set $P_B = 2D$. Let $D^+$ denote the generalized inverse, or the moore penrose inverse of $D$. (Since $B$ is positive semi-definite, $D^{-1}$ might not be well-defined). We then set $\mu_B = - D^+b$. Since $B\succeq 0$, using lemma \ref{lemm:schur_comp}.3, we have: 
\begin{itemize}
    \item $D\succeq 0$
    \item $(I-DD^+)b = 0 \implies DD^+b = b \implies -\frac 12 P_B\mu_B = b$
\end{itemize}

We move on to show that this transformation again gives identical models. 
\begin{align}
   (\tilde Z_x-\mu_B)^\top P_B (\tilde Z_x-\mu_B) &= \tilde Z_x^\top P_B \tilde Z_x - 2 \tilde Z_x^\top P_B \mu_B +\mu_B^\top P_B\mu_B\\
   & = 2\tilde Z_x^\top D \tilde Z_x +  \tilde Z_x^\top b + 2b^\top D^+D D^+b\\
   & = 2 Z_x^\top BZ_x +  2b^\top  D^+b - 2a
\end{align}
Thus, 
\begin{align*}
    f(x|B) & = \frac{1}{c(B)} \exp\left(-Z_x^\top B Z_x\right)\\
    &= \frac{e^{b^\top D^+b-a}}{c(B)} \exp\left(-\frac 12  (\tilde Z_x-\mu_B)^\top P_B (\tilde Z_x-\mu_B)\right)\\
\end{align*}

Thus, with this $(\mu_B,P_B)$, the  representation \ref{eq:sn_pdf_old}  gives the same model as the sliced normal in \ref{eq:sn_pdf}, where $c(\mu,P) =  c(B)/e^{b^\top D^+b-a}$. 
\end{proof}

\begin{proof}{Proof of Theorem \ref{thm: CVXOPT}}
We show that the maximum–likelihood search problem
\begin{equation}
\label{eq:MLE_problem_updated}
    \min_{B\succ0}\;
        J(B)\;:=\;
        \frac1n\sum_{i=1}^n Z_{x_i}^{\top}B Z_{x_i}
        \;+\;
        \log\Bigl(
              c(B)
            \Bigr),
        \qquad
        c(B):=\int_{\mathcal X}
               \exp\bigl(-Z_x^{\top}B Z_x\bigr)\,dx,
\end{equation}
is convex in the decision variable
$B\in\bbS_{++}^d$ (the cone of real, symmetric, positive–definite
$d\times d$ matrices).  Since $\bbS_{++}^d$ is itself convex, it
suffices to prove that $J$ is a convex function.

\paragraph{1. Affinity of the empirical term.}
Writing the first summand in trace form,
\[
  \frac1n\sum_{i=1}^n Z_{x_i}^{\top}B Z_{x_i}
  \;=\;
  \tr\Bigl(
      B\,\underbrace{\Bigl(\tfrac1n\sum_{i=1}^n
                           Z_{x_i}Z_{x_i}^{\top}\Bigr)}_{=:S}
     \Bigr),
\]
which is affine (hence convex) in $B$.

\paragraph{2. Convexity of the log‑partition term.}
Set $A(x):=Z_xZ_x^{\top}\succeq0$ and define
\[
  \phi(B)
    :=\log c(B)
    =\log\Bigl(
        \int_{\mathcal X}
        \exp\bigl(-\tr(BA(x))\bigr)\,dx
      \Bigr).
\]
Take any $B_1,B_2\in\bbS_{++}^d$ and $\theta\in(0,1)$.
Then
\begin{align*}
  \exp\bigl(\phi(\theta B_1+(1-\theta)B_2)\bigr)
  &=\int_{\mathcal X}
      \exp\Bigl(
          -\tr\bigl((\theta B_1+(1-\theta)B_2)A(x)\bigr)
        \Bigr)\,dx\\
  &=\int_{\mathcal X}
      \bigl[\exp(-\tr(B_1A(x)))\bigr]^{\theta}\,
      \bigl[\exp(-\tr(B_2A(x)))\bigr]^{1-\theta}\,dx.
\end{align*}
Applying Hölder’s inequality with conjugate exponents
$p=1/\theta$ and $q=1/(1-\theta)$ yields
\[
  \exp\bigl(\phi(\theta B_1+(1-\theta)B_2)\bigr)
  \;\le\;
  \Bigl(\int_{\mathcal X}\exp(-\tr(B_1A(x)))\,dx\Bigr)^{\theta}
  \Bigl(\int_{\mathcal X}\exp(-\tr(B_2A(x)))\,dx\Bigr)^{1-\theta}.
\]
Taking logarithms gives
\[
  \phi(\theta B_1+(1-\theta)B_2)
  \;\le\;
  \theta\,\phi(B_1)+(1-\theta)\,\phi(B_2),
\]
so $\phi$ is convex on $\bbS_{++}^d$.

\paragraph{3. Convexity of the objective.}
The objective $J(B)$ is the sum of the affine term from Step~1 and the convex function $\phi(B)$ from Step~2; hence $J$ is convex. 
\end{proof}

\end{document}